\documentclass[runningheads,orivec,envcountsame]{llncs}

\usepackage{centernot}
\usepackage{amssymb}
\usepackage{algorithm2e}
\usepackage{stmaryrd}
\usepackage{bm}
\usepackage{tikz}
\usetikzlibrary{shapes,calc,arrows,automata}
\usepackage{multirow}
\usepackage{array}
\usepackage{mathtools}
\usepackage{enumitem}
\usepackage[bookmarks,unicode,colorlinks=true]{hyperref}
\usepackage{proof}
\usepackage[capitalize,nameinlink]{cleveref}
\usepackage[location=appendix,manual]{moveproofs}
\usepackage{thm-restate}
\usepackage{subfig}
\usepackage{hhline}
\usepackage{stackengine, graphicx}
\usepackage{cite}
\usepackage{microtype}
\usepackage[font=small]{caption}
\usepackage{xspace}
\usepackage{marvosym}

\everypar{\looseness=-1}
\allowdisplaybreaks[4]
\setlist{itemsep=1pt, topsep=2pt}

\everydisplay=\expandafter{\the\everydisplay\small}

\DontPrintSemicolon

\hypersetup{
  pdftitle={Proving Non-Termination by Acceleration Driven Clause Learning},
  pdfauthor={Florian Frohn and Jürgen Giesl},
  colorlinks=true,
  linkcolor=blue,
  citecolor=olive,
  filecolor=magenta,
  urlcolor=cyan
}

\makeatletter
\RequirePackage[bookmarks,unicode,colorlinks=true]{hyperref}%
   \def\@citecolor{blue}%
   \def\@urlcolor{blue}%
   \def\@linkcolor{blue}%

\def\orcidID#1{\smash{\href{http://orcid.org/#1}{\protect\raisebox{-1.25pt}{\protect\includegraphics{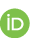}}}}}
\makeatother

\newenvironment{proofsketch}{%
  \proof}{\endproof}

\DeclareMathAlphabet{\mathpzc}{OT1}{pzc}{m}{it}

\renewcommand{\epsilon}{\varepsilon}
\let\oldphi\phi
\let\oldvarphi\varphi
\renewcommand{\varphi}{\oldphi}
\renewcommand{\phi}{\oldvarphi}

\newcommand{\nt}{\mathsf{nt}}
\newcommand{\pred}[1]{\mathsf{#1}}
\newcommand{\pl}[1]{\textsf{#1}}
\newcommand{\tool}[1]{\textsf{#1}}
\newcommand{\inc}[1]{#1\raisebox{.4ex}{\tiny\bf++}}
\newcommand{\dec}[1]{#1\raisebox{.4ex}{\tiny\bf--\! --}}
\newcommand{\eq}[1]{\overset{=}{#1}}

\newcommand{\Init}{\textsc{Init}\xspace}
\newcommand{\Step}{\textsc{Step}\xspace}
\newcommand{\Accelerate}{\textsc{Accelerate}\xspace}
\newcommand{\Covered}{\textsc{Covered}\xspace}
\newcommand{\Backtrack}{\textsc{Backtrack}\xspace}
\newcommand{\Refute}{\textsc{Refute}\xspace}
\newcommand{\Prove}{\textsc{Prove}\xspace}
\newcommand{\Nonterm}{\textsc{Nonterm}\xspace}

\newcommand{\proj}[2]{\langle #1 \rangle_{#2}}
\newcommand{\mbp}{\mathsf{sip}}

\newcommand{\unsafe}{\mathsf{unsafe}}
\newcommand{\safe}{\mathsf{safe}}
\newcommand{\bt}{\mathsf{bt}}
\newcommand{\cond}{\mathsf{cond}}
\newcommand{\guard}[1]{\,\llbracket #1 \rrbracket}
\newcommand{\accel}{\mathsf{accel}}
\newcommand{\chain}{\mathsf{chain}}
\newcommand{\init}{\mathsf{init}}
\newcommand{\err}{\mathsf{err}}
\newcommand{\QF}{\mathsf{QF}}
\renewcommand{\AA}{\mathcal{A}}
\newcommand{\LL}{\mathcal{L}}
\newcommand{\VV}{\mathcal{V}}
\newcommand{\state}[1]{\ensuremath{\mathfrak{#1}}}

\newcommand{\ZZ}{\mathbb{Z}}

\newcommand{\NN}{\mathbb{N}}
\newcommand{\CC}{\mathcal{C}}

\newcommand{\TT}{\mathcal{T}}

\renewcommand{\SS}{\mathcal{S}}

\newcommand{\Def}{\mathrel{\mathop:}=}

\renewcommand{\emptyset}{\varnothing}

\crefname{equation}{eq.}{equations}%
\crefname{chapter}{chapter}{chapters}%
\crefname{section}{sect.}{sections}%
\crefname{appendix}{app.}{appendices}%
\crefname{enumi}{item}{items}%
\crefname{footnote}{footnote}{footnotes}%
\crefname{figure}{fig.}{figures}%
\crefname{table}{table}{tables}%
\crefname{theorem}{thm.}{theorems}%
\crefname{lemma}{lemma}{lemmas}%
\crefname{corollary}{cor.}{corollaries}%
\crefname{proposition}{proposition}{propositions}%
\crefname{definition}{def.}{definitions}%
\crefname{result}{result}{results}%
\crefname{example}{ex.}{examples}%
\crefname{remark}{remark}{remarks}%
\crefname{note}{note}{notes}%

\newcommand{\report}[1]{#1}
\newcommand{\paper}[1]{}

\title{Proving Non-Termination by Acceleration Driven Clause Learning\paper{ (Short Paper)}\thanks{funded by
    the Deutsche Forschungsgemeinschaft (DFG, German Research Foundation)
    - 235950644 (Project GI 274/6-2)}}
\titlerunning{Proving Non-Termination by Acceleration Driven Clause Learning}
\author{Florian Frohn\paper{$^{(\mbox{\Letter})}$}\orcidID{0000-0003-0902-1994} \and Jürgen Giesl\paper{$^{(\mbox{\Letter})}$}\orcidID{0000-0003-0283-8520}}
\institute{LuFG Informatik 2, RWTH Aachen University, Aachen, Germany}
\authorrunning{F.\ Frohn, J.\ Giesl}

\begin{document}

\maketitle

\begin{abstract}
  We recently proposed \emph{Acceleration Driven Clause Learning} (ADCL), a
  novel calculus to analyze satisfiability of \emph{Con\-strained Horn Clauses} (CHCs).
  Here, we adapt ADCL to transition systems and introduce ADCL-NT, a variant for disproving termination.
  We implemented ADCL-NT in our tool \tool{LoAT} and evaluate it against the state of the art.
\end{abstract}

\section{Introduction}
\label{sect:Introduction}

Termination is one of the most important properties of programs, and thus termination
analysis is a very active field of research.
Here, we are concerned with \emph{dis}proving termination of \emph{transition systems} (TSs), a popular intermediate representation for verification of programs written in more expressive languages.

\begin{example}
  \label{ex:leading}
  Consider the following TS $\TT$ with entry-point $\init$ and two further \emph{locations} $\ell_1,\ell_2$ over the variables $x,y,z$, where $x',y',z'$ represent the values of $x,y,z$ \emph{after} applying a transition, and $\eq{x}, \inc{x}$, and $\dec{x}$ abbreviate $x' = x$, $x' = x + 1$, and $x' = x - 1$.
  The first two transitions are a variant\footnote{We generalized the example to make it more interesting, and we added the condition $y \leq 2 \cdot z$ to enforce termination
  of $\tau_{\ell_1}$.}
    of {\tt chc-LIA-Lin\_052} from the \emph{CHC Competition~'22}
    \cite{CHC-COMP} and the last two are a variant\footnote{We combined the
    transitions for the cases $x > y$ and $x < y$ into the equivalent transition
    \ref{eq:ex1-neq} to demonstrate how our approach can deal with disjunctions in
    conditions.} of {\tt flip2\_rec.jar-obl-8} from the \emph{Termination and Complexity Competition
    (TermComp)} \cite{termcomp}.
  \begin{align}
    \init & {} \to \ell_1 \guard{x' \leq 0 \land z' \geq 5000 \land y' \leq z'}  \label{eq:ex1-init} \tag{\protect{\ensuremath{\tau_{\mathsf{i}}}}} \\
    \ell_1 & {} \to \ell_1 \guard{y \leq 2 \cdot z \land \inc{x} \land ((x < z \land \eq{y}) \lor (x \geq z \land \inc{y})) \land \eq{z}} \label{eq:ex1-rec} \tag{\protect{\ensuremath{\tau_{\ell_1}}}} \\
    \ell_1 & {} \to \ell_2 \guard{x = y \land x > 2 \cdot z \land \eq{x} \land \eq{y}} \label{eq:ex1-nonrec} \tag{\protect{\ensuremath{\tau_{\mathsf{\ell_1 \to \ell_2}}}}} \\
    \ell_2 & {} \to \ell_2 \guard{x = y \land x > 0 \land \eq{x} \land \dec{y}} \label{eq:ex1-eq} \tag{\protect{\ensuremath{\tau_{\ell_2}^=}}}\\
    \ell_2 & {} \to \ell_2 \guard{x > 0 \land y > 0 \land x'=y \land ((x > y \land y' = x) \lor (x < y \land \eq{y}))} \label{eq:ex1-neq} \tag{\protect{\ensuremath{\tau_{\ell_2}^{\neq}}}}
  \end{align}
  At $\ell_1$, $\TT$ operates in two ``phases'':
  First, just $x$ is incremented until $x$ reaches $z$ ($1^{st}$ disjunct of \ref{eq:ex1-rec}).
  Then, $x$ and $y$ are incremented until $y$ reaches $2 \cdot z + 1$ ($2^{nd}$ disjunct of
  \ref{eq:ex1-rec}).
  If $x = y = c$ holds for some $c > 1$ at that point (which is the case if $x \leq y = z$ holds initially), then the execution can continue at $\ell_2$ as follows:
  \begin{gather*}
    \ell_2(c,c,c_z) \longrightarrow_{\ref{eq:ex1-eq}} \ell_2(c,c-1,c_z) \longrightarrow_{\ref{eq:ex1-neq}} \ell_2(c-1,c,c_z) \longrightarrow_{\ref{eq:ex1-neq}} \ell_2(c,c,c_z) \longrightarrow_{\ref{eq:ex1-eq}} \ldots
  \end{gather*}
  Here, $\ell_2(c,c,c_z)$ means that the current location is $\ell_2$ and the values of $x,y,z$ are $c,c,c_z$.
  The $1^{st}$ and $2^{nd}$ step with \ref{eq:ex1-neq}
  satisfy the $1^{st}$ ($x > y \land \ldots$) and $2^{nd}$ ($x < y \land \ldots$) disjunct of \ref{eq:ex1-neq}'s condition, respectively.
  Thus, $\TT$ does not terminate.
\end{example}

\Cref{ex:leading} is challenging for state-of-the-art tools for several reasons.
First, more than 5000
steps are required to reach $\ell_2$, so reachability of $\ell_2$ is difficult to prove for approaches that unroll the transition relation or use other variants of iterative deepening.
Thus, {\tt chc-LIA-Lin\_052} is beyond the capabilities of most other
state-of-the-art tools for proving reachability.

Second, the pattern ``$\ref{eq:ex1-eq}$, $1^{st}$ disjunct of $\ref{eq:ex1-neq}$, $2^{nd}$ disjunct of $\ref{eq:ex1-neq}$'' must be found to prove non-termination.
Therefore, {\tt flip2\_rec.jar-obl-8} (which does not use disjunctions) cannot be solved by
other state-of-the-art termination tools.

Third, \Cref{ex:leading} contains disjunctions, which are not supported by many termination tools.
Presumably, the reason is that most techniques for (dis)proving termination of loops are restricted to conjunctions (e.g., due to the use of templates and Farkas' Lemma).
While disjunctions can be avoided by splitting disjunctive transitions according to the DNF of their conditions, this leads to an exponential blow-up
in the number of transitions.

We present an approach that can prove non-termination of systems
like \Cref{ex:leading} automatically.
To this end, we tightly integrate non-termination techniques into our
recent \emph{Acceleration Driven Clause Learning (ADCL)}
calculus \cite{cav23}, which has originally been designed for CHCs, but it can also be used to analyze TSs.

Due to the use of acceleration techniques that compute the transitive closure of recursive
transitions, ADCL finds long witnesses of reachability automatically.
If acceleration techniques cannot be applied, it unrolls the transition relation, so it can easily detect complex patterns of transitions that admit non-terminating runs.
Finally, ADCL reduces reasoning about disjunctions to reasoning about conjunctions by
considering conjunctive variants of disjunctive transitions.
Thus, combining ADCL with
non-termination techniques for
conjunctive transitions allows for disproving termination of TSs with complex Boolean
structure.

After introducing preliminaries in \Cref{sec:preliminaries}, \Cref{sec:ADCL} presents a straightforward adaption of ADCL to TSs.
\Cref{sec:nonterm} introduces our main contribution: ADCL-NT, a variant of ADCL for proving non-termination.
Finally, in \Cref{sec:experiments}, we discuss related work and demonstrate the power of our approach
by comparing it with other state-of-the-art tools.
All proofs can be found in\report{ the appendix}\paper{ \cite{report}}.

\section{Preliminaries}
\label{sec:preliminaries}

We assume familiarity with basics from many-sorted first-order logic.
$\VV$ is a countably infinite set of variables and $\AA$ is a first-order theory over a $k$-sorted signature $\Sigma_\AA$ with carrier $\CC_\AA = (\CC_{\AA,1},\ldots,\CC_{\AA,k})$.
$\QF(\Sigma_\AA)$ is the set of all
quantifier-free first-order formulas over $\Sigma_\AA$, which are  w.l.o.g.\ assumed to be in negation normal form, and $\QF_\land(\Sigma_\AA)$ only contains conjunctions of $\Sigma_\AA$-literals.
Given a first-order formula $\eta$ over $\Sigma_\AA$, $\sigma$ is a \emph{model} of
$\eta$ (written $\sigma \models_\AA \eta$) if it is a model of $\AA$ with carrier
$\CC_\AA$, extended with interpretations for $\VV$ such that $\eta$ is satisfied.
As usual, $\models_\AA \eta$ means that $\eta$ is valid,
and $\eta \equiv_\AA \eta'$ means $\models_\AA \eta \iff \eta'$.

We write $\vec{x}$ for sequences and $x_i$ is the $i^{th}$ element of $\vec{x}$.
We use ``$::$'' for concatenation of sequences, where we identify sequences of length $1$ with their elements, so we may write, e.g., $x::\mathit{xs}$ instead of $[x]::\mathit{xs}$.

\vspace{-0.5em}
\paragraph*{\bf Transition Systems}

Let $d \in \NN$ be fixed, and let $\vec{x},\vec{x}' \in \VV^d$ be disjoint vectors of pairwise different variables.
Each $\psi \in \QF(\Sigma_\AA)$ induces a relation $\longrightarrow_\psi$ on $\CC_\AA^d$ where $\vec{s} \longrightarrow_\psi \vec{t}$ iff $\psi[\vec{x}/\vec{s},\vec{x}'/\vec{t}]$ is satisfiable.
So for the condition $\psi \Def ({x = y} \land {x > 0} \land \eq{x} \land \dec{y})$
 of \ref{eq:ex1-eq}, we have $(4,4,4)  \longrightarrow_{\psi} (4,3,7)$.
$\LL \supseteq \{\init,\err\}$ is a finite set of \emph{locations}.
A \emph{configuration} is a pair $(\ell,\vec{s}) \in \LL \times \CC_\AA^d$, written $\ell(\vec{s})$.
A \emph{transition} is a triple $\tau = (\ell,\psi,\ell') \in \LL \times \QF(\Sigma_\AA)
\times \LL$, written $\ell \to \ell' \guard{\psi}$, and its \emph{condition} is
$\cond(\tau) \Def \psi$.
W.l.o.g., we assume $\ell \neq \err$ and $\ell' \neq \init$.
Then $\tau$ induces a relation $\longrightarrow_{\tau}$ on configurations where $\state{s} \longrightarrow_{\tau} \state{t}$ iff $\state{s} = \ell(\vec{s}), \state{t} = \ell'(\vec{t})$, and $\vec{s} \longrightarrow_\psi \vec{t}$.
So, e.g., $\ell_2(4,4,4) \longrightarrow_{\tau_{\ell_2}^=} \ell_2(4,3,7)$.
We call $\tau$ \emph{recursive} if $\ell = \ell'$, \emph{conjunctive} if $\psi \in \QF_\land(\Sigma_\AA)$, \emph{initial} if $\ell = \init$, and \emph{safe} if $\ell' \neq \err$.
Moreover, we define $(\ell \to \ell' \guard{\psi})|_{\psi'} \Def \ell \to \ell' \guard{\psi'}$.
A \emph{transition system} (TS) $\TT$ is a finite set of transitions, and it induces the relation $\longrightarrow_{\TT} \Def \bigcup_{\tau \in \TT} {\longrightarrow_{\tau}}$.

\emph{Chaining} $\tau = \ell_s \to \ell_t \guard{\psi}$ and $\tau' = \ell_s' \to \ell_t' \guard{\psi'}$ yields $\chain(\tau,\tau') \Def (\ell_s \to \ell'_t \guard{\psi_c})$ where $\psi_c \Def \psi[\vec{x}' / \vec{x}''] \land \psi'[\vec{x} / \vec{x}'']$ for fresh $\vec{x}'' \in \VV^d$ if $\ell_t = \ell_s'$, and $\psi_c \Def \bot$ (meaning $\mathit{false}$) if $\ell_t \neq \ell_s'$.
So ${\longrightarrow_{\chain(\tau,\tau')}} = {\longrightarrow_\tau} \circ {\longrightarrow_{\tau'}}$, and $\chain(\ref{eq:ex1-nonrec},\ref{eq:ex1-eq}) = \ell_1
\to \ell_2 \guard{\psi}$ where $\psi \equiv_\AA (x = y \land x > 2
\cdot z \land x > 0 \land \eq{x} \land \dec{y})$.
For non-empty, finite sequences of transitions we define $\chain([\tau]) \Def \tau$ and
$\chain([\tau_1,\tau_2]::\vec{\tau}) \Def \chain(\chain(\tau_1,\tau_2)::\vec{\tau})$.
We lift notations for transitions to finite sequences via chaining.
So $\cond(\vec{\tau}) \Def \cond(\chain(\vec{\tau}))$,
$\vec{\tau}$ is \emph{recursive} if $\chain(\vec{\tau})$ is recursive,
${\longrightarrow_{\vec{\tau}}} = {\longrightarrow_{\chain(\vec{\tau})}}$, etc.
If $\tau$ is initial and $\cond(\tau::\vec{\tau}) \not\equiv_\AA \bot$, then $(\tau::\vec{\tau}) \in \TT^+$ is a \emph{finite run}.
$\TT$ is safe if every finite run is safe.
If there is a $\sigma$ such that $\sigma \models \cond(\vec{\tau}')$ for every finite prefix $\vec{\tau}'$ of $\vec{\tau} \in \TT^\omega$, then $\vec{\tau}$ is an \emph{infinite run}.
If no infinite run exists, then $\TT$ is \emph{terminating}.

\vspace{-0.5em}
\paragraph*{\bf Acceleration}

\emph{Acceleration techniques} compute the transitive closure of relations.
In the following definition, we only consider relations defined by conjunctive formulas, since many existing acceleration techniques do not support disjunctions \cite{bozga09a}, or have to resort to approximations in the presence of disjunctions \cite{acceleration-calculus}.

\begin{definition}[Acceleration]
  \label{def:accel}
  An \emph{acceleration technique} is a function $\accel: \QF_\land(\Sigma_\AA) \mapsto
  \QF_\land(\Sigma_{\AA'})$ such that ${\longrightarrow_{\psi}^+} =
     {\longrightarrow_{\accel(\psi)}}$,
     where $\AA'$ is a first-order theory.
  For recursive conjunctive transitions $\tau$, we define $\accel(\tau) \Def \tau|_{\accel(\cond(\tau))}$.
\end{definition}
So we clearly have ${\longrightarrow^+_\tau} = {\longrightarrow_{\accel(\tau)}}$.
Note that most theories are not ``closed under acceleration''.
E.g., accelerating the Presburger formula $x'_1 = x_1 + x_2 \land \eq{x_2}$ yields the non-linear formula $n > 0 \land x'_1 = x_1 + n \cdot x_2 \land \eq{x_2}$.
If neither $\NN$ nor $\ZZ$ are contained in $\CC_\AA$, then an additional sort for the range of $n$ is required in the formula that results from applying $\accel$.
Hence, \Cref{def:accel} allows $\AA' \neq \AA$.

\section{ADCL for Transition Systems}
\label{sec:ADCL}

We originally proposed the ADCL calculus to analyze satisfiability
of linear \emph{Con\-strained Horn Clauses} (CHCs)
\cite{cav23}. Here, we rephrase it for TSs, and in \Cref{sec:nonterm}, we modify it for
proving non-termination.
The adaption to TSs is straightforward as TSs can be transformed into
equivalent linear CHCs and vice versa (see, e.g., \cite{samir19}).

To bridge the gap between transitions $\tau$ where $\cond(\tau) \in \QF(\Sigma_\AA)$ and acceleration techniques for formulas from $\QF_\land(\Sigma_\AA)$, ADCL uses \emph{syntactic implicants}.

\begin{definition}[Syntactic Implicants \protect{\cite[Def.\ 6]{cav23}}]
  \label{def:implicant}
  If $\psi \in \QF(\Sigma_\AA)$, then:
  \begin{align*}
    \mbp(\psi,\sigma) & {} \Def \bigwedge \{\pi \text{ is a literal of } \psi \mid \sigma \models_\AA \pi\} && \text{if $\sigma \models_\AA \psi$} \\
    \mbp(\psi) & {} \Def \{ \mbp(\psi,\sigma) \mid \sigma \models_\AA \psi \}\\
    \mbp(\tau) & {} \Def \{\tau|_{\psi} \mid \psi \in \mbp(\cond(\tau))\} && \text{for transitions $\tau$} \\
    \mbp(\TT) & {} \Def \bigcup_{\tau \in \TT} \mbp(\tau) && \text{for TSs $\TT$}
  \end{align*}
  Here, $\mbp$ abbreviates \emph{syntactic implicant projection}.
\end{definition}
As $\mbp(\psi,\sigma)$ is restricted to literals from $\psi$, $\mbp(\psi)$ is finite.
Syntactic implicants ignore the semantics of literals.
So we have, e.g., $(X > 1) \notin \mbp(X > 0 \land X > 1) = \{X > 0 \land X > 1\}$.
It is easy to show $\psi \equiv_\AA \bigvee \mbp(\psi)$, and thus ${\longrightarrow_\TT} = {\longrightarrow_{\mbp(\TT)}}$.

Since $\mbp(\tau)$ is worst-case exponential in the size of $\cond(\tau)$, we do not compute it explicitly.
Instead, ADCL constructs a run $\vec{\tau}$ step by step, and to perform a step with $\tau$, it searches for a model $\sigma$ of $\cond(\vec{\tau}::\tau)$.
If such a model exists, it appends $\tau|_{\mbp(\cond(\tau),\sigma)}$ to $\vec{\tau}$.
This corresponds to a step with a conjunctive variant of $\tau$ whose condition is satisfied by $\sigma$.
In other words, our calculus constructs $\mbp(\cond(\tau), \sigma)$ ``on the fly'' when
performing a step with $\tau$, where $\sigma \models_\AA \cond(\vec{\tau}::\tau)$

The core idea of ADCL is to learn new, \emph{non-redundant} transitions via acceleration.
Essentially, a transition is redundant if its transition relation is a subset of another transition's relation.
Thus, redundant transitions are not useful for (dis-)proving safety.

\begin{definition}[Redundancy, \protect{\cite[Def.\ 8]{cav23}}]
  \label{def:redundancy}
  A transition $\tau$ is \emph{(strictly) redundant} w.r.t.\ $\tau'$, denoted $\tau \sqsubseteq \tau'$ ($\tau \sqsubset \tau'$) if ${\longrightarrow_\tau} \subseteq {\longrightarrow_{\tau'}}$ (${\longrightarrow_\tau} \subset {\longrightarrow_{\tau'}}$).
  For a TS $\TT$, we have $\tau \sqsubseteq \TT$ ($\tau \sqsubset \TT$) if $\tau \sqsubseteq \tau'$ ($\tau \sqsubset \tau'$) for some $\tau' \in \TT$.
\end{definition}
In the sequel, we assume oracles for redundancy, satisfiability of $\QF(\Sigma_\AA)$-formulas, and acceleration.
In practice, we use incomplete techniques instead (see \Cref{sec:experiments}).

From now on, let $\TT$ be the TS that is being analyzed with ADCL.
A \emph{state} of ADCL consists of a TS $\SS$ that augments $\TT$ with
\emph{learned transitions}, a run $\vec{\tau}$ of $\SS$ called the
\emph{trace},
and a sequence of sets of \emph{blocking transitions} $[B_i]_{i=0}^k$,
where transitions that are redundant w.r.t.\ $B_k$ must not be appended to the trace.

The following definition introduces the ADCL calculus.
It extends the trace step by step (using the rule \Step, which performs an evaluation step with a transition) and learns new transitions via acceleration (\Accelerate) whenever a suffix of the trace is recursive.
To avoid non-terminating ADCL-derivations, our notion of \emph{redundancy} from
\Cref{def:redundancy} is used to backtrack whenever a suffix of the trace corresponds to a special case of another (learned) transition (\Covered).
Moreover, \Backtrack is used whenever a run cannot be continued.
A more detailed explanation of ADCL is provided after \Cref{def:calc}.

\begin{definition}[ADCL \protect{\cite[Def.\ 9, 10]{cav23}}]
  \label{def:state}
  \label{def:calc}
  A \emph{state} is a triple $(\SS,[\tau_i]_{i=1}^k,[B_i]_{i=0}^k)$
  where $\SS \supseteq \TT$ is a TS, $\bigcup_{i=0}^k B_i \subseteq \mbp(\SS)$, and $[\tau_i]_{i=1}^k \in \mbp(\SS)^*$.
  The transitions in $\mbp(\TT)$ are called \emph{original} and the transitions in $\mbp(\SS)\setminus\mbp(\TT)$ are \emph{learned}.
  A transition $\tau_{k+1} \sqsubseteq B_k$ is \emph{blocked}, and 
  $\tau_{k+1} \not\sqsubseteq B_k$ is \emph{active} if $\chain([\tau_i]_{i=1}^{k+1})$ is
  an initial transition with satisfiable condition (i.e., $[\tau_i]_{i=1}^{k+1}$ is a run). Let
  \begin{gather*}
    \bt(\SS,[\tau_i]_{i=1}^k,[B_0,\ldots,B_k]) \Def (\SS,[\tau_i]_{i=1}^{k-1},[B_0,\ldots,B_{k-1} \cup \{\tau_k\}])
  \end{gather*}
  where $\bt$ abbreviates ``backtrack''.
  Our calculus is defined by the following rules.
  \[
    \begin{array}{cr@{\quad}cr}
      \infer{\TT \leadsto (\TT,[],[\emptyset])}{} & (\Init) & \infer{(\SS,\vec{\tau},\vec{B}) \leadsto (\SS,\vec{\tau}::\tau,\vec{B}::\emptyset)}{\tau \in \mbp(\SS) \text{ is active}} & (\Step)
      \\[0.5em]
      \multicolumn{3}{c}{\infer{(\SS,\vec{\tau}::\vec{\tau}^\circlearrowleft,\vec{B}::\vec{B}^\circlearrowleft) \leadsto (\SS \cup \{\tau\},\vec{\tau}::\tau,\vec{B}::\{\tau\})}{\vec{\tau}^\circlearrowleft \text{ is recursive} & |\vec{\tau}^\circlearrowleft| = |\vec{B}^\circlearrowleft| & \accel(\vec{\tau}^\circlearrowleft) = \tau \not\sqsubseteq \mbp(\SS)}} & \mathllap{(\Accelerate)} \\[0.5em]
      \multicolumn{3}{c}{\infer{s = (\SS,\vec{\tau}::\vec{\tau}',\vec{B}) \leadsto \bt(s)}{\vec{\tau}' \sqsubset \mbp(\SS) \qquad \text{or} \qquad \vec{\tau}' \sqsubseteq \mbp(\SS) \land |\vec{\tau}'| > 1}} & (\Covered) \\[0.5em]
      \multicolumn{3}{c}{\infer{s = (\SS,\vec{\tau}::\tau,\vec{B}) \leadsto \bt(s)}{\text{all transitions from $\mbp(\SS)$ are inactive} & \quad \tau \text{ is safe}}} & \mathllap{(\Backtrack)} \\[0.5em]
      \infer{(\SS,\vec{\tau},\vec{B}) \leadsto \unsafe}{\vec{\tau} \text{ is unsafe}} & \quad (\Refute) & \infer{(\SS,[],[B]) \leadsto \safe}{\text{all transitions from $\mbp(\SS)$ are inactive}} & (\Prove)
    \end{array}
  \]
  \vspace{-1em}
\end{definition}
We write $\overset{\textsc{I}}{\leadsto}$, $\overset{\textsc{S}}{\leadsto}$, $\ldots$ to indicate that the rule \Init, {\Step}, $\ldots$ was used.
{\Step} adds a transition to the trace.
When the trace has a recursive suffix, \Accelerate allows for learning a new transition which then replaces the recursive suffix on the trace, or we may backtrack via \Covered if the recursive suffix is redundant.
Note that \Covered does not apply if $\vec{\tau}' \sqsubseteq \mbp(\SS)$ and $|\vec{\tau}'| = 1$, as it could immediately undo every \Step, otherwise.
If no further {\Step} is possible, \Backtrack applies.
Note that \Backtrack and \Covered block the last transition from the trace so that we do not perform the same {\Step} again.
If $\vec{\tau}$ is an unsafe run, \Refute yields $\unsafe$, and if the entire
search space has been exhausted without finding an unsafe run (i.e., if all initial transitions are blocked), \Prove yields $\safe$.

The definition of ADCL in \cite{cav23} is more liberal than ours:
In our setting, \Accelerate may only be applied if the learned transition is non-redundant, and our definition of ``active transitions'' enforces that the first transition on the trace is always an initial transition.
In \cite{cav23}, these requirements are not enforced by the definition of ADCL, but by the definition of \emph{reasonable strategies} \cite[Def.\ 14]{cav23}.
For simplicity, we integrated these requirements into \Cref{def:calc}.
Additionally, \Covered should be preferred over
\Accelerate, and \Accelerate should be preferred over
{\Step}.

\begin{example}
  \label{ex:unsafe}
  We apply ADCL to a version of \Cref{ex:leading} with the additional transition
  \begin{equation}
    \tag{\protect{\ensuremath{\tau_{\err}}}}
    \label{eq:ex-err}
    \ell_1 \to \err \guard{x = y \land x > 2 \cdot z \land \eq{x} \land \eq{y} \land \eq{z}}.
  \end{equation}
  \scalebox{0.95}{
    \parbox{\textwidth}{
      \vspace{-1em}
      \begin{align*}
        \TT \overset{\textsc{I}}{\leadsto}^{\phantom{2}} {} & (\TT,[],[\emptyset]) \overset{\textsc{S}}{\leadsto}^2 (\TT,[\ref{eq:ex1-init},\ref{eq:ex1-rec}|_{\psi_{x < z}}],[\emptyset,\emptyset,\emptyset]) \tag{$x \leq 1 \land z \geq 5k \land y \leq z$} \\
        {} \overset{\textsc{A}}{\leadsto}^{\phantom{2}} {} & (\SS_1,[\ref{eq:ex1-init},\tau_{x<z}^+],[\emptyset,\emptyset,\{\tau_{x<z}^+\}]) \tag{$x \leq z \land z \geq 5k \land y \leq z$} \\
        {} \overset{\textsc{S}}{\leadsto}^{\phantom{2}} {} & (\SS_1,[\ref{eq:ex1-init},\tau_{x<z}^+,\ref{eq:ex1-rec}|_{\psi_{x \geq z}}],[\emptyset,\emptyset,\{\tau_{x<z}^+\},\emptyset]) \tag{$x = z + 1 \land z \geq 5k \land y \leq z + 1$} \\
        {} \overset{\textsc{A}}{\leadsto}^{\phantom{2}} {} & (\SS_2,[\ref{eq:ex1-init},\tau_{x<z}^+,\tau_{x \geq z}^+],[\emptyset,\emptyset,\{\tau_{x<z}^+\},\{\tau_{x \geq z}^+\}]) \tag{$x \geq y \land x > z \geq 5k \land y \leq 2 \cdot z + 1$} \\
        {} \overset{\textsc{S}}{\leadsto}^{\phantom{2}} {} & (\SS_2,[\ref{eq:ex1-init},\tau_{x<z}^+,\tau_{x \geq z}^+,\ref{eq:ex-err}],[\emptyset,\emptyset,\{\tau_{x<z}^+\},\{\tau_{x \geq z}^+\},\emptyset]) \tag{$x = 2 \cdot z + 1 = y \land z \geq 5k$} \\
        {} \overset{\textsc{R}}{\leadsto}^{\phantom{2}} {} & \unsafe
      \end{align*}
      \vspace{-1em}
    }}

  \noindent
  Here, $5k$ abbreviates $5000$ and:
  \begin{align*}
    \psi_{x < z} & {} \Def y \leq 2 \cdot z \land \inc{x} \land x < z \land \eq{y} \land \eq{z} & \psi_{x \geq z} & {} \Def y \leq 2 \cdot z \land \inc{x} \land x \geq z \land \inc{y} \land \eq{z} \\
    \tau_{x<z}^+ & {} \Def \mathrlap{\ell_1 \to \ell_1 \guard{y \leq 2 \cdot z \land n > 0 \land x' = x + n \land x + n \leq z \land \eq{y} \land \eq{z}}} \\
    \tau_{x \geq z}^+ & {} \Def \mathrlap{\ell_1 \to \ell_1 \guard{y + n - 1 \leq 2 \cdot z \land n > 0 \land x' = x + n \land x \geq z \land y' = y + n \land \eq{z}}} \\
    \SS_1 & {} \Def \TT \cup \{\tau_{x<z}^+\} & \SS_2 & {} \Def \SS_1 \cup \{\tau_{x \geq z}^+\}
  \end{align*}

  \noindent
  On the right, we show formulas describing the configurations that are reachable with the current trace.
  Every $\leadsto$-derivation starts with \Init.
  The first two {\Step}s add the initial transition \ref{eq:ex1-init} and an element of $\mbp(\ref{eq:ex1-rec})$ to the trace.
  Since $x < z$ holds after applying \ref{eq:ex1-init}, the only possible choice for the latter is $\ref{eq:ex1-rec}|_{\psi_{x < z}}$.

  As $\ref{eq:ex1-rec}|_{\psi_{x < z}}$ is recursive, it is accelerated and replaced with $\accel(\ref{eq:ex1-rec}|_{\psi_{x < z}}) = \tau^+_{x<z}$, which simulates $n$ steps with $\ref{eq:ex1-rec}|_{\psi_{x < z}}$.
  Moreover, $\tau_{x<z}^+$ is also added to the current set of blocking transitions, as we
  always have ${\longrightarrow^2_{\tau}} \subseteq {\longrightarrow_{\tau}}$ for learned
  transitions $\tau$ and thus adding them to the trace twice in a row is pointless.

  Next, \ref{eq:ex1-rec} is applicable again.
  As neither $x < z$ nor $x \geq z$ holds for all reachable configurations, we could continue with any element of $\mbp(\ref{eq:ex1-rec}) = \{\ref{eq:ex1-rec}|_{\psi_{x < z}},\ref{eq:ex1-rec}|_{\psi_{x \geq z}}\}$.
  We choose $\ref{eq:ex1-rec}|_{\psi_{x \geq z}}$, so that the recursive transition $\ref{eq:ex1-rec}|_{\psi_{x \geq z}}$ can be accelerated to $\tau^+_{x \geq z}$.
  Then \ref{eq:ex-err} applies, and the proof is finished via \Refute.
\end{example}
For our purposes, the most important property of ADCL is the following.
\begin{restatable}{theorem}{correct}
  \label{thm:correct}
  If $\TT \leadsto^* (\SS,\vec{\tau},\vec{B})$ and $\vec{\tau}$ is non-empty, then  $\cond(\vec{\tau}) \not \equiv_\AA \bot$ and ${\longrightarrow_{\vec{\tau}}} \subseteq {\longrightarrow^+_\TT}$.
  So if $\TT \leadsto^* \unsafe$, then $\TT$ is unsafe.
\end{restatable}
\makeproof*{thm:correct}{
  \correct*
  \begin{proof}
    Analogously to the corresponding proof from \cite{cav23}, $\TT \leadsto^* (\SS,\vec{\tau},\vec{B})$ implies ${\longrightarrow^+_\TT} = {\longrightarrow^+_\SS}$.
    Moreover, \Cref{def:calc} ensures $\vec{\tau} \in \SS^*$.
    Furthermore, satisfiability of $\cond(\vec{\tau})$ is an invariant of $\leadsto$.
    Finally, the definition of \Step ensures that the first element of $\vec{\tau}$ is an initial transition.
    Thus, the claim follows. \qed
  \end{proof}
}

The other properties of ADCL that were shown in \cite{cav23} immediately carry over to our setting, too:
if $\TT \leadsto^* \safe$, then $\TT$ is safe;
if $\TT$ is unsafe, then $\TT \leadsto^* \unsafe$;
in general, $\leadsto$ does not terminate.
The proofs are analogous to \cite{cav23}.

\section{Proving Non-Termination with ADCL-NT}
\label{sec:nonterm}

From now on, we assume that the analyzed TS $\TT$ does not contain unsafe transitions.
To prove non-termination, we look for a corresponding \emph{certificate}.

\begin{definition}[Certificate of Non-Termination]
  \label{def:certificate}
  Let $\tau = \ell \to \ell \guard{\ldots}$.
  A satisfiable formula $\psi$ \emph{certifies non-ter\-mi\-na\-tion of $\tau$}, written
  $\psi \models_\AA^\infty \tau$, if for any model $\sigma$ of $\psi$,
  there is an infinite sequence
  $
    \ell(\sigma(\vec{x})) = \state{s}_1 \longrightarrow_\tau \state{s}_2 \longrightarrow_\tau \ldots
  $
\end{definition}
There exist many techniques for finding certificates of non-termination automatically, see \Cref{sec:experiments}.
However, \Cref{def:certificate} has several shortcomings.
First, the problem of finding such certificates becomes very challenging if $\cond(\tau)$ contains disjunctions.
Second, it is insufficient to consider a single transition when only non-singleton sequences $\vec{\tau}$ such that $\chain(\vec{\tau})$ is recursive admit non-terminating runs.
Third, just finding a certificate $\psi$ of non-termination for some $\vec{\tau} \in \TT^*$ does
not suffice for proving non-termination of $\TT$.
Additionally, a proof that the pre-image of $\longrightarrow_{\vec{\tau}|_\psi}$ is reachable from an initial configuration is required.
All of these problems can be solved by integrating the search for certificates of non-termination into the ADCL calculus.

\begin{definition}[ADCL-NT]
  \label{def:calc-nonterm}
  To prove non-termination, we extend ADCL with the rule \Nonterm and modify \Covered as shown below.
  We write $\leadsto_\nt$ for the relation defined by the (modified) rules from \Cref{def:calc} and \Nonterm.
  \[
  \begin{array}{c@{\qquad}r}
    \infer{
      s = (\SS,\vec{\tau}::\vec{\tau}^\circlearrowleft,\vec{B}) \leadsto_\nt \bt(s)
      }{\vec{\tau}^\circlearrowleft \text{ is recursive} & \quad \vec{\tau}^\circlearrowleft \sqsubset \mbp(\SS) \text{ or } \vec{\tau}^\circlearrowleft \sqsubseteq \mbp(\SS) \land |\vec{\tau}^\circlearrowleft| > 1} & (\Covered) \\[0.5em]
    \infer{
      (\SS,\vec{\tau}::\vec{\tau}^\circlearrowleft,\vec{B}) \leadsto_\nt (\SS \cup \{\tau\}, \vec{\tau}::\vec{\tau}^\circlearrowleft,\vec{B})
      }{\chain(\vec{\tau}^\circlearrowleft) = \ell \to \ell \guard{\ldots} & \psi \models_\AA^\infty \vec{\tau}^\circlearrowleft & \tau = \ell \to \err \guard{\psi} \not\sqsubseteq \mbp(\SS)} & (\Nonterm)
  \end{array}
  \]
\end{definition}
So the idea of \Nonterm is to apply a technique which searches for a certificate of non-termination to a recursive suffix of the trace.
Apart from introducing \Nonterm, we restricted \Covered to recursive suffixes.
The reason is that backtracking when the trace has a redundant, non-recursive suffix may prevent us from analyzing loops, resulting in a precision issue.
\begin{example}
  Let $\TT \Def \{\tau_{\mathsf{i}}, \tau'_{\mathsf{i}}, \tau_\ell,\tau_{\ell'}\}$ where
  \begin{gather*}
    \tau_{\mathsf{i}} \Def \init \to \ell \guard{\top} \quad \tau_{\mathsf{i}}' \Def \init \to \ell' \guard{\top} \quad
    \tau_\ell \Def \ell \to \ell' \guard{\top} \quad \tau_{\ell'} \Def \ell' \to \ell \guard{\top}
  \end{gather*}
  and $\top$ means \emph{true}.
  Due to the loop $\ell \longrightarrow_{\tau_\ell} \ell' \longrightarrow_{\tau_{\ell'}} \ell$, $\TT$ is clearly non-terminating.
  Without requiring that $\vec{\tau}^\circlearrowleft$ is recursive in \Covered, $\TT$ can be analyzed as follows:
  \begin{align*}
    \TT & {} \overset{\textsc{I}}{\leadsto}_\nt (\TT,[],[\emptyset]) \overset{\textsc{S}}{\leadsto}_\nt^2 (\TT,[\tau_{\mathsf{i}},\tau_\ell],\emptyset^3) \overset{\textsc{C}}{\leadsto}_\nt (\TT,[\tau_{\mathsf{i}}],[\emptyset,\{\tau_\ell\}]) \overset{\textsc{B}}{\leadsto}_\nt (\TT,[],[\{\tau_{\mathsf{i}}\}])
    \\
        & {} \overset{\textsc{S}}{\leadsto}_\nt^2 (\TT,[\tau_{\mathsf{i}}',\tau_{\ell'}],\{\tau_{\mathsf{i}}\} ::\emptyset^2) \overset{\textsc{C}}{\leadsto}_\nt (\TT,[\tau_{\mathsf{i}}'],[\{\tau_{\mathsf{i}}\},\{\tau_{\ell'}\}]) \overset{\textsc{B}}{\leadsto}_\nt (\TT,[],[\{\tau_{\mathsf{i}},\tau'_{\mathsf{i}}\}]) \overset{\textsc{P}}{\leadsto}_\nt \safe
  \end{align*}
  The $1^{st}$ application of \Covered is possible as $[\tau_{\mathsf{i}},\tau_\ell] \sqsubseteq \tau_{\mathsf{i}}'$ and the $2^{nd}$ application of \Covered is possible as $[\tau'_{\mathsf{i}},\tau_{\ell'}] \sqsubseteq \tau_{\mathsf{i}}$.
  Note that the trace never contains both $\tau_\ell$ and $\tau_{\ell'}$, but both transitions are needed to prove non-termination.
\end{example}

Recall the shortcomings of \Cref{def:certificate} mentioned above.
First, due to the use of syntactic implicants, ADCL-NT reduces reasoning about arbitrary transitions to reasoning about conjunctive transitions.
Second, as \Nonterm considers a suffix $\vec{\tau}^\circlearrowleft$ of the trace, it can prove non-termination of sequences of transitions.
Third, ADCL's capability to prove reachability directly carries over to our goal of
proving non-termination.
So in contrast to most other approaches (see \Cref{sec:experiments}), ADCL-NT does not have to resort to other tools or techniques for proving reachability.

We only search for a certificate of non-termination for $\vec{\tau}^\circlearrowleft$ if ADCL-NT established reachability of the pre-image of $\longrightarrow_{\vec{\tau}^\circlearrowleft}$ beforehand.
Note, however, that this does not imply reachability of the pre-image of
$\longrightarrow_{\ell \to \err \guard{\psi}}$, as $\psi$ entails $\cond(\vec{\tau}^\circlearrowleft)$, but
not the other way around.
Hence, we cannot directly derive non-termination of $\TT$ when \Nonterm applies.
Regarding the strategy for $\leadsto_\nt$, one should try to use \Nonterm once for each recursive suffix of the trace.

\begin{example}
  \label{ex:nonterm}
  Reconsider \Cref{ex:leading}.
  Up to (excluding) the second-last step, the der\-i\-va\-tion from \Cref{ex:unsafe} remains unchanged.
  Then we get
  \begin{align*}
    & (\SS_2,[\ref{eq:ex1-init},\tau_{x<z}^+,\tau_{x \geq z}^+],[\dots]) \tag{$x \geq y \land x > 5k$} \\
    {} \overset{\textsc{S}}{\leadsto}^4_\nt
    {} & (\SS_2,[\ref{eq:ex1-init},\tau_{x<z}^+,\tau_{x \geq z}^+,\ref{eq:ex1-nonrec},\ref{eq:ex1-eq},\ref{eq:ex1-neq}|_{\psi_{x>y}},\ref{eq:ex1-neq}|_{\psi_{x<y}}],[\ldots]) \tag{$1 \equiv_2 y = x > 10k$} \\
    {} \overset{\textsc{N}}{\leadsto}_\nt {} & (\SS_3,[\ref{eq:ex1-init},\tau_{x<z}^+,\tau_{x \geq z}^+,\ref{eq:ex1-nonrec},\ref{eq:ex1-eq},\ref{eq:ex1-neq}|_{\psi_{x>y}},\ref{eq:ex1-neq}|_{\psi_{x<y}}],[\ldots]) \tag{$1 \equiv_2 y = x > 10k$} \\
    {} \overset{\textsc{S}}{\leadsto}_\nt {} & (\SS_3,[\ref{eq:ex1-init},\tau_{x<z}^+,\tau_{x \geq z}^+,\ref{eq:ex1-nonrec},\ref{eq:ex1-eq},\ref{eq:ex1-neq}|_{\psi_{x>y}},\ref{eq:ex1-neq}|_{\psi_{x<y}},\tau_\err],[\ldots]) \overset{\textsc{R}}{\leadsto}_\nt {} \unsafe
  \end{align*}
  \vspace{-2em}
  \begin{align*}
    \text{where} \quad \psi_{x>y} & {} \Def x > 0 \land y > 0 \land x' = y \land x > y \land y' = x & \tau_\err & {} \Def \ell_2 \to \err \guard{x = y > 1} \\
    \psi_{x<y} & {} \Def x > 0 \land y > 0 \land x' = y \land x < y \land \eq{y} & \SS_3 & {} \Def \SS_2 \cup \{\tau_\err\}
  \end{align*}
The formulas on the right describe the values of $x$ and $y$ that are reachable with the current trace, where $1 \equiv_2 y$ means that $y$ is odd.
  After the first \Step with \ref{eq:ex1-nonrec}, just \ref{eq:ex1-eq} can be used, as $\cond(\ref{eq:ex1-nonrec})$ implies $x' = y'$.
  While \ref{eq:ex1-eq} is recursive, \Accelerate cannot be applied next, as ${{\longrightarrow_{\ref{eq:ex1-eq}}} = {\longrightarrow_{\ref{eq:ex1-eq}}^+}}$, so the learned transition would be redundant.
  Thus, we continue with \ref{eq:ex1-neq}, projected to $x>y$ (as $\cond(\ref{eq:ex1-eq})$ implies $x' = y'+1$).
  Again, all transitions that could be learned are redundant, so \Accelerate does not apply.
  We next use \ref{eq:ex1-neq} projected to $x<y$, as the previous \Step swapped $x$ and $y$.
  As the suffix
  $[\ref{eq:ex1-eq},\ref{eq:ex1-neq}|_{\psi_{x>y}},\ref{eq:ex1-neq}|_{\psi_{x<y}}]$ of the
  trace does not terminate (see \Cref{ex:leading}), \Nonterm applies.
  So we learn the transition $\tau_\err$, which is added to the trace to finish the proof, afterwards.
\end{example}

\vspace*{-.2cm}

\begin{restatable}{theorem}{correctnt}
  \label{thm:sound-nt}
  If $\TT \leadsto_\nt^* \unsafe$, then $\TT$ does not terminate.
\end{restatable}
\makeproof*{thm:sound-nt}{
  \correctnt*
  \begin{proof}
    We have
    \begin{gather*}
      \TT \leadsto_\nt^* (\SS,\vec{\tau},\vec{B}) \overset{\textsc{S}}{\leadsto}_\nt (\SS,\vec{\tau}::(\ell \to \err \guard{\psi}),\vec{B}::\emptyset) \overset{\textsc{R}}{\leadsto}_\nt \unsafe.
    \end{gather*}
    Then by \Cref{lem:run} and the definition of \Step, there are  $\vec{s},\vec{t} \in \CC_\AA^d$ such that $\init(\vec{s}) \longrightarrow_{\TT}^* \ell(\vec{t})$ and $\psi[\vec{x} / \vec{t}]$ is satisfiable.

    As we assumed that $\TT$ does not contain any transitions to $\err$, it follows that $\ell \to \err \guard{\psi}$ was learned via \Nonterm.
    Hence, $\psi$ is a certificate of non-ter\-mi\-na\-tion for some $\vec{\tau}^\circlearrowleft \in \SS^+$ such that $\chain(\vec{\tau}^\circlearrowleft) = \ell \to \ell \guard{\ldots}$.
    Thus, by \Cref{def:certificate}, we have
    \begin{gather*}
      \init(\vec{s}) \longrightarrow^*_\TT \ell(\vec{t}) \longrightarrow^+_{\SS} \ell(\vec{t}_1) \longrightarrow^+_{\SS} \ell(\vec{t}_2) \longrightarrow^+_{\SS} \ldots
    \end{gather*}
    As ${\longrightarrow^+_{\SS}} = {\longrightarrow^+_{\TT}}$, the claim follows. \qed
  \end{proof}
}

While \Cref{thm:sound-nt} establishes the soundness of our approach, we now investigate completeness.
In contrast to ADCL for safety (\Cref{sec:ADCL}), ADCL-NT is not refutationally complete,
but the proof is non-trivial.
So in the following, we show that there are non-terminating TSs $\TT$ where $\TT \not\leadsto^*_\nt
\unsafe$.
To prove incompleteness, we adapt the construction from the proof that ADCL does not terminate \cite[Thm.~18]{cav23}.
There, states $(\SS, \vec{\tau}, \vec{B})$ were extended by a component
$\LL$ that maps every element of $\mbp(\SS)$ to a regular language over $\mbp(\TT)$.
However, the proof of \cite[Thm.~18]{cav23} just required reasoning about finite (prefixes of infinite) runs, but we have to reason about infinite runs.
So in our setting $\LL$ maps each element $\tau$ of $\mbp(\SS)$ to a regular or an
$\omega$-regular language over $\mbp(\TT)$,
i.e.,  $\LL(\tau) \subseteq \mbp(\TT)^*$ or 
$\LL(\tau) \subseteq \mbp(\TT)^\omega$.
We lift $\LL$ from $\mbp(\SS)$ to  sequences of transitions as follows.
\begin{gather*}
  \LL(\epsilon) \Def \epsilon \qquad \qquad \qquad \LL(\vec{\tau}::\tau) \Def \LL(\vec{\tau})::\LL(\tau) \quad \text{if} \quad \LL(\tau) \subseteq \mbp(\tau)^*
\end{gather*}
Here, ``$::$'' denotes language concatenation (i.e., $\LL_1 :: \LL_2 = \{ \tau_1 :: \tau_2
\mid 
\tau_1 \in \LL_1, \tau_2 \in \LL_2 \}$)
and
we only consider sequences where $\LL(\tau)$
is regular (not $\omega$-regular)
to ensure that $\LL$ is well defined.
So while we lift other notations to sequences of transitions via
chaining, $\LL(\vec{\tau})$ does \emph{not} stand for $\LL(\chain(\vec{\tau}))$.

\def\scale{0.97}
\begin{definition}[ADCL-NT with Regular Languages]
  \label{def:calc-reg}
  We extend states by a fourth component $\LL$, and adapt \Init, \Accelerate, and \Nonterm as
  follows:
  \[
    \begin{array}{cr}
      \scalebox{\scale}{
      \infer{
      \TT \leadsto_\nt (\TT,[],[\emptyset],\LL)
      }{\LL(\tau) = \{\tau\} \text{ for all $\tau \in \mbp(\TT)$}}
      } & (\Init)\\[0.5em]
      \scalebox{\scale}{
      \infer{
      (\SS,\vec{\tau}::\vec{\tau}^\circlearrowleft,\vec{B}::\vec{B}^\circlearrowleft,\LL) \leadsto_\nt (\SS \cup \{\tau\},\vec{\tau}::\tau,\vec{B}::\{\tau\},\LL \uplus (\tau \mapsto \LL(\vec{\tau}^\circlearrowleft)^+))
      }{\vec{\tau}^\circlearrowleft \text{ is recursive} \qquad & |\vec{\tau}^\circlearrowleft| = |\vec{B}^\circlearrowleft| \qquad & \accel(\vec{\tau}^\circlearrowleft) = \tau \not\sqsubseteq \mbp(\SS)}} & (\Accelerate)\\[0.5em]
      \scalebox{\scale}{
      \infer{
      (\SS,\vec{\tau}::\vec{\tau}^\circlearrowleft,\vec{B},\LL) \leadsto_\nt (\SS \cup \{\tau\}, \vec{\tau}::\vec{\tau}^\circlearrowleft,\vec{B},\LL \uplus (\tau \mapsto \LL(\vec{\tau}^\circlearrowleft)^\omega))
      }{\chain(\vec{\tau}^\circlearrowleft) = \ell \to \ell \guard{\ldots} & \psi \models_\AA^\infty \vec{\tau}^\circlearrowleft & \tau = \ell \to \err \guard{\psi} \not\sqsubseteq \mbp(\SS)}} & (\Nonterm)
    \end{array}
  \]
  All other rules from \Cref{def:calc} leave the last component of the state unchanged.
\end{definition}
Here, $\LL(\pi)^+ \Def \bigcup_{n \in \NN_{\geq 1}} \LL(\pi)^n$, and $\LL(\pi)^\omega$ is the $\omega$-regular language consisting of all words that result from concatenating infinitely many elements of $\LL(\pi) \setminus \{\epsilon\}$.

In \Accelerate and \Nonterm,
$\chain(\vec{\tau}^\circlearrowleft)$ is recursive.
Thus, $\vec{\tau}^\circlearrowleft$ does not contain unsafe transitions.
Hence, $\LL(\vec{\tau}^\circlearrowleft)$
and thus also $\LL(\vec{\tau}^\circlearrowleft)^+$
 are well defined and regular, and
$\LL(\vec{\tau}^\circlearrowleft)^\omega$ is $\omega$-regular.
Moreover, the use of ``$\uplus$''
is justified by the condition $\tau \not\sqsubseteq \mbp(\SS)$.
The next lemma states two crucial properties about $\LL$.

\begin{restatable}{lemma}{lang}
  \label{lem:lang}
  Assume  $\TT \leadsto_\nt^* (\SS,\vec{\tau},\vec{B},\LL)$ and let $\tau = (\ell \to \ell' \guard{\psi}) \in \mbp(\SS)$.
  \begin{itemize}
  \item[$\bullet$] If $\LL(\tau) \subseteq \mbp(\TT)^*$, then ${\longrightarrow_{\tau}} = \bigcup_{\vec{\tau} \in \LL(\tau)}{\longrightarrow_{\vec{\tau}}}$.
  \item[$\bullet$] If $\LL(\tau) \subseteq \mbp(\TT)^\omega$, then
for every model $\sigma$ of $\psi$, there is an infinite sequence
     $\ell(\sigma(\vec{x})) = \state{s}_1 \longrightarrow_{\tau_1} \state{s}_2
\longrightarrow_{\tau_2} \ldots$ where  $[\tau_1, \tau_2, \ldots]
\in  \LL(\tau)$.
  \end{itemize}
\end{restatable}
\makeproof*{lem:lang}{
  \lang*
  \begin{proof}
    The first part of the lemma is analogous to \cite[Lemma 17]{cav23}.
    For the second part of the lemma, let $\SS' \Def \SS \setminus \{\tau\}$ and ${\longrightarrow_{\LL(\tau)}} \Def \bigcup_{\vec{\tau} \in \LL(\tau)}{\longrightarrow_{\vec{\tau}}}$.
    Clearly, $\tau$ has been learned via \Nonterm.
    Hence, there is a $\vec{\tau}^\circlearrowleft = [\tau^\circlearrowleft_1,\ldots,\tau^\circlearrowleft_k] \in (\SS')^*$ such that $\psi \models_{\AA}^\infty \chain(\vec{\tau}^\circlearrowleft)$, i.e., for every model $\sigma$ of $\psi$, there is an infinite sequence
    \begin{gather*}
      \ell(\sigma(\vec{x})) = \state{t}_0 \longrightarrow_{\chain(\vec{\tau}^\circlearrowleft)} \state{t}_1 \longrightarrow_{\chain(\vec{\tau}^\circlearrowleft)} \ldots
    \end{gather*}
    and thus there is also an infinite sequence
    \begin{gather*}
      \ell(\sigma(\vec{x})) = \state{t}_{0} = \state{t}_{0,0} \longrightarrow_{\tau_1^\circlearrowleft} \state{t}_{0,1} \longrightarrow_{\tau_2^\circlearrowleft} \ldots \longrightarrow_{\tau_k^\circlearrowleft} \state{t}_{0,k} = \state{t}_1 = \state{t}_{1,0} \longrightarrow_{\tau_1^\circlearrowleft} \ldots
    \end{gather*}
    Let $\sigma \models_\AA \psi$ be arbitrary but fixed.
 Then by the first part of the lemma, we get
    \begin{gather*}
      \ell(\sigma(\vec{x})) = \state{t}_{0,0} \longrightarrow_{\LL(\tau_1^\circlearrowleft)} \state{t}_{0,1} \longrightarrow_{\LL(\tau_2^\circlearrowleft)} \ldots \longrightarrow_{\LL(\tau_k^\circlearrowleft)} \state{t}_{0,k} = \state{t}_{1,0} \longrightarrow_{\LL(\tau_1^\circlearrowleft)} \ldots
      \end{gather*}
      Therefore, we also have
    \begin{gather*}
      \ell(\sigma(\vec{x})) = \state{t}_{0,0} \longrightarrow_{\vec{\tau}_{0,1}} \state{t}_{0,1} \longrightarrow_{\vec{\tau}_{0,2}} \ldots \longrightarrow_{\vec{\tau}_{0,k}} \state{t}_{0,k} = \state{t}_{1,0} \longrightarrow_{\vec{\tau}_{0,1}} \ldots
    \end{gather*}
    where $\vec{\tau}_{i,j} \in \LL(\tau_j^\circlearrowleft)$ for all $i \in \NN$ and all $1 \leq j \leq k$.
    As argued before, all transitions in $\vec{\tau}^\circlearrowleft$ are safe because $\chain(\vec{\tau}^\circlearrowleft)$ is recursive.
    Thus, $\LL$ maps all transitions in $\vec{\tau}^\circlearrowleft$ to regular languages.
    Therefore, by the definition of the lifting of $\LL$ to sequences, we get $\vec{\tau}_{i,1}::\ldots::\vec{\tau}_{i,k} \in \LL(\vec{\tau}^\circlearrowleft)$ for all $i \in \NN$.
    Hence, there exists an infinite sequence $[\tau_1, \tau_2, \ldots]
    \in \LL(\vec{\tau}^\circlearrowleft)^\omega = \LL(\tau)$ with
     \begin{gather*}
      \ell(\sigma(\vec{x})) = \state{s}_1 \longrightarrow_{\tau_1} \state{s}_2 \longrightarrow_{\tau_2} \ldots
    \end{gather*}\qed
  \end{proof}
}
\noindent
Based on this lemma, we can prove that our extension of $\leadsto_\nt$ from \Cref{def:calc-reg} is not refutationally complete.
Then refutational incompleteness of ADCL-NT as introduced in \Cref{def:calc-nonterm} follows immediately.
The reason is that $\LL$ is only used in the premise of \Init in \Cref{def:calc-reg}, but there the requirement ``$\LL(\tau) = \{\tau\}$ for all $\tau \in \mbp(\TT)$'' is trivially satisfiable by choosing $\LL$ accordingly.
\begin{restatable}{theorem}{incomplete}
  \label{thm:incomplete}
  There is a non-terminating TS $\TT$ such that $\TT \not\leadsto_\nt^* \unsafe$.
\end{restatable}
\begin{proofsketch}
  As in the proof of \cite[Thm.~18]{cav23},
  for any (original or learned) transition $\tau$ such that $\LL(\tau)$ is regular,
  $\LL(\tau)$ contains at most one square-free word (i.e., a word without a non-empty infix $w::w$).
  Thus, if $\LL(\tau)$ is $\omega$-regular, then $\LL(\tau)$ does not contain an infinite square-free word.
  Moreover, as in the proof of \cite[Thm.~18]{cav23}, one can construct a TS $\TT$ that admits a single infinite run $\vec{\tau}$,
  and this infinite run is square-free.
  Thus, there is no transition $\tau$ such that $\LL(\tau)$ contains a suffix of $\vec{\tau}$, i.e.,
  no $\leadsto_\nt$-derivation starting with $\TT$ corresponds to $\vec{\tau}$.
  Hence, by \Cref{lem:lang}, assuming $\TT \leadsto^*_\nt \unsafe$ results in a contradiction. \qed
\end{proofsketch}
\makeproof*{thm:incomplete}{
  \incomplete*
  \begin{proof}
    Similar to the proof of \cite[Thm.\ 18]{cav23}, one can construct a non-terminating TS $\TT$ over the theory $\AA_{LIA}$ of linear integer arithmetic such that every run corresponds to a \emph{square-free word}, i.e., a word without a non-empty infix of the form $w::w$.
    We first recapitulate the construction from \cite[Thm.\ 18]{cav23}.

    We consider the Thue-Morse sequence $[v_i]_{i \in \NN}$ \cite{oeis-thue-morse}.
    Let $w_i \Def v_{i+1} - v_i$.
    The resulting infinite sequence $[w_i]_{i \in \NN}$ over the alphabet $\{-1,0,1\}$ is well-known to be square-free \cite{oeis-thue-morse-difference}.
    Then $\TT$ contains
    \begin{equation}
      \label{eq:nt1}
      \tag{\protect{\ensuremath{\tau_{\pred{init}}}}}
      \init \to \pred{ThueMorse} \guard{i=0 \land x=1 \land \eq{i} \land \eq{x}},
    \end{equation}
    the following transitions $\TT_{\pred{ThueMorse}}$
    \begin{align*}
      \pred{ThueMorse} & {} \to \pred{next} \guard{x = -1 \land \eq{x} \land \inc{i}}  \\
      \pred{ThueMorse} & {} \to \pred{next} \guard{x = 0 \land \eq{x} \land \inc{i}}  \\
      \pred{ThueMorse} & {} \to \pred{next} \guard{x = 1 \land \eq{x} \land \inc{i}},
    \end{align*}
    and transitions $\TT_{\pred{next}}$ such that ${\longrightarrow_{\TT_{\pred{next}}}}$
    is well founded and
    \begin{gather*}
      \pred{next}(i,x) \longrightarrow_{\TT_{\pred{next}}} \ldots \longrightarrow_{\TT_{\pred{next}}} \pred{ThueMorse}(i',x') \quad \text{iff} \quad i>0 \land \eq{i} \land x = w_{i-1} \land x' = w_i.
    \end{gather*}
    Note that $\TT_{\pred{next}}$ exists, since $[w_i]_{i \in \NN}$ is computable and TSs over $\AA_{LIA}$ are Turing complete.
    Moreover, we may assume that the $\longrightarrow_{\TT_{\pred{next}}}$-sequence above is unique for each $i > 0$.
    The reason is that \emph{deterministic} TSs over $\AA_{LIA}$ (where each configuration has at most one successor) are still Turing complete.
    Thus, $\TT$ admits one and only one infinite run.
    W.l.o.g., we assume $\mbp(\TT) = \TT$ in the sequel.
    Then
    \begin{multline}
      \label{eq:all-square-free}
      \text{for the unique infinite sequence $\init(0,1) = \state{s}_1 \longrightarrow_{\tau_1} \state{s}_2 \longrightarrow_{\tau_2} \ldots$}\\
      \text{where $\tau_1,\tau_2,\ldots \in \TT$, the word $\proj{[\tau_1,\tau_2,\ldots]}{\TT_{\pred{ThueMorse}}}$  is square-free}
    \end{multline}
    due to square-freeness of $[w_i]_{i \in \NN}$ and termination of $\TT_{\pred{next}}$.
    Here, the notation $\proj{[\tau_1,\tau_2,\ldots]}{\TT_{\pred{ThueMorse}}}$ denotes the sequence that results from $[\tau_1,\tau_2,\ldots]$ by omitting all transitions that are not contained in $\TT_{\pred{ThueMorse}}$.

    Assume $\TT \leadsto_\nt^+ (\SS,\vec{\tau},\vec{B},\LL) \leadsto_\nt \unsafe$.
    Then $\proj{\LL(\tau)}{\TT_{\pred{ThueMorse}}}$ does not contain an infinite square-free word, for each $\tau \in \SS$.
    To see this, first note that $\LL(\tau)$ is built from singleton languages over finite words, concatenation, Kleene plus, and the operator  $\{ ... \}^\omega$.
    Thus, $\proj{\LL(\tau)}{\TT_{\pred{ThueMorse}}}$ is built using the same operations.

    We first prove that any language $\LL$ over a finite alphabet that is built from singleton languages, concatenation, and Kleene plus contains at most one square-free word.
    We use induction on the construction of $\LL$.
    If $\LL$ is a singleton language, then the claim is trivial.
    If $\LL = \LL'::\LL''$, then the square-free words in $\LL$ are a subset of
    \begin{equation}
      \label{eq:square-free}
      \{w'::w'' \mid w' \in \LL', w'' \in \LL'', w' \text{ and } w'' \text{ are square-free}\}.
    \end{equation}
    By the induction hypothesis, there is at most one square-free word $w' \in \LL'$ and at most one square-free word $w'' \in \LL''$.
    Thus, the size of \eqref{eq:square-free} is at most one.
    If $\LL = (\LL')^+$, then $\LL$ contains the same square-free words as $\LL'$.
    To see this, let $w \in (\LL')^+$ be square-free.
    Then there are $n \in \NN_{\geq 1}$, $v_1,\ldots,v_n \in \LL'$ such that $v_1::\ldots::v_n = w$.
    Since $w$ is square-free, each $v_i$ must be square-free, too.
    As $\LL'$ contains at most one square-free word by the induction hypothesis, we get
    $v_1 = \ldots = v_n$.
    Since $w$ is square-free, this implies $n = 1$ (or that $w$ and the $v_i$ are the empty word).

    Next, note that $\LL^\omega$ cannot contain an infinite square-free word if $\LL$ contains at most one square-free word.
    To see this, let $v \in \LL^\omega$ be square-free, where $v = v_0::v_1::\ldots$ for
    $v_0,v_1,\ldots \in \LL$.
    All $v_i$ must be square-free, because otherwise $v$ would not be square-free either.
    As $\LL$ contains at most one square-free word, we get $v_i=v_j$ for all $i,j \in \NN$.
    If the only square-free word in $\LL$ is $\epsilon$, then $v =\epsilon \notin \LL^\omega$.
    Otherwise, $v$ is not square-free.

    Hence, it follows that
    \begin{equation}
      \label{eq:not-square-free}
      \text{$\proj{\LL(\tau)}{\TT_{\pred{ThueMorse}}}$ contains no infinite square-free word}.
    \end{equation}

    Assume
    \begin{gather*}
      \TT \leadsto_\nt^+ (\SS,\vec{\tau},\vec{B},\LL) \overset{\textsc{S}}{\leadsto}_\nt (\SS,\vec{\tau}::(\ell \to \err \guard{\psi}),\vec{B},\LL) \overset{\textsc{R}}{\leadsto}_\nt \unsafe.
    \end{gather*}
    So by \Cref{lem:run} and the definition of \Step, we have $\init(0,1) \longrightarrow_\TT^+ \ell(c_i,c_x)$ where $\psi[i/c_i,x/c_x]$ is satisfiable.
    As $\TT$ is safe, we have $\LL(\ell \to \err \guard{\psi}) \subseteq \TT^\omega$ by the definition of \Nonterm.
    Therefore, by \Cref{lem:lang}, there is a $[\tau_1,\tau_2,\ldots] \in \LL(\ell \to \err \guard{\psi})$ such that
    \begin{gather*}
      \ell(c_i,c_x) = \state{s}_1 \longrightarrow_{\tau_1} \state{s}_2 \longrightarrow_{\tau_2} \ldots
    \end{gather*}
    Thus, $\init(0,1) \longrightarrow_{\TT}^+ \state{s}_1 \longrightarrow_{\tau_1} \state{s}_2 \longrightarrow_{\tau_2} \ldots$ is an infinite $\longrightarrow_\TT$-derivation.
    However, as $(\ell \to \err \guard{\psi}) \in \SS$, the word $\proj{[\tau_1,\tau_2,\ldots]}{\TT_{\pred{ThueMorse}}}$
    is not square-free due to \eqref{eq:not-square-free}.
    Thus, for any prefix $\vec{\tau}$, the word $\proj{\vec{\tau} :: [\tau_1,\tau_2,\ldots]}{\TT_{\pred{ThueMorse}}}$ is not square-free either,
    contradicting \eqref{eq:all-square-free}.
    So we contradicted the assumption $\TT \leadsto_\nt^+ \unsafe$ and hence the theorem follows. \qed
  \end{proof}
}

Since ADCL can prove unsafety as well as safety, it is natural to ask if there is a dual to ADCL-NT that can prove termination.
The most obvious approach would be the following:
Whenever the trace has a recursive suffix $\vec{\tau}^\circlearrowleft$, then termination
of $\vec{\tau}^\circlearrowleft$ needs to be proven before the next
$\leadsto$-step.
The following example shows that this is not enough to ensure that $\TT \leadsto_\nt^+ \safe$
implies termination of $\TT$.

\begin{example}
  Let $\TT \Def \{\tau_{\mathsf{i}} = \init \to \ell \guard{\psi_{\mathsf{i}}}\} \cup \{\tau_m = \ell \to \ell \guard{\psi_m} \mid 0 \leq m \leq 2\}$ and
  \begin{align*}
     \psi_{\mathsf{i}} & {} \Def {x' = 0} & \psi_0 & {} \Def {x=0 \land x'=1} & \psi_1 & {} \Def {x=1 \land x'=2} & \psi_2 \Def {x=2 \land x'=1}.
  \end{align*}
  As we have $\ell(1) \longrightarrow_{\tau_1} \ell(2) \longrightarrow_{\tau_2} \ell(1)$, $\TT$ is clearly non-terminating.
  We get:

  \noindent
  \begin{align*}
    \TT \overset{\textsc{I}}{\leadsto}_\nt {} & (\TT,[],[\emptyset]) \overset{\textsc{S}}{\leadsto}_\nt^3 (\TT,[\tau_{\mathsf{i}},\tau_0,\tau_1],\emptyset^4) \overset{\textsc{A}}{\leadsto}_\nt (\SS_1,[\tau_{\mathsf{i}},\tau_{01}],\emptyset^2::\{\tau_{01}\}) \\
    {} \overset{\textsc{S}}{\leadsto}_\nt {} & (\SS_1,[\tau_{\mathsf{i}},\tau_{01},\tau_2],\emptyset^2::\{\tau_{01}\}::\emptyset) \overset{\textsc{A}}{\leadsto}_\nt (\SS_2,[\tau_{\mathsf{i}},\tau_{012}],\emptyset^2::\{\tau_{01},\tau_{012}\}) \\
    {} \overset{\textsc{S}}{\leadsto}_\nt {} & (\SS_2,[\tau_{\mathsf{i}},\tau_{012},\tau_1],\emptyset^2::\{\tau_{01},\tau_{012}\}::\emptyset) \overset{\textsc{C}}{\leadsto}_\nt (\SS_2,[\tau_{\mathsf{i}},\tau_{012}],\emptyset^2::\{\tau_{01},\tau_{012},\tau_1\}) \\
    {} \overset{\textsc{B}}{\leadsto}_\nt {} & (\SS_2,[\tau_{\mathsf{i}}],\emptyset::\{\tau_{012}\}) \leadsto_\nt^* (\SS_2,[\tau_{\mathsf{i}}],\emptyset::\{\tau_{012},\tau_0,\tau_{01}\}) \overset{\textsc{B}}{\leadsto}_\nt (\SS_2,[],[\{\tau_{\mathsf{i}}\}]) \overset{\textsc{P}}{\leadsto}_\nt \safe
  \end{align*}

\noindent
After three {\Step}s, we accelerate the recursive suffix $[\tau_0,\tau_1]$ of the
trace, resulting in $\tau_{01} = \ell \to \ell \guard{x = 0 \land x' = 2}$ and $\SS_1 =
\TT \cup \{\tau_{01}\}$.
  After one more step, $[\tau_{01},\tau_2]$ is accelerated to $\tau_{012} = \ell \to \ell
  \guard{x = 0 \land x' = 1}$ and we get $\SS_2 = \SS_1 \cup \{ \tau_{012} \}$.
  After the next step, $[\tau_{012},\tau_1]$ is redundant w.r.t.\ $\tau_{01}$, so \Covered applies.
  Then we \Backtrack, as no other transitions are active.
  The next {\Step}s also yield states that allow for backtracking (as their traces have the redundant suffixes $[\tau_0,\tau_1]$ and $[\tau_{01},\tau_2]$), so we can finally apply \Backtrack again and finish with \Prove.

  Note that whenever the trace has a recursive suffix, then it leads from $\ell(i)$ to $\ell(j)$ where $i \neq j$, i.e., each such suffix is trivially terminating.
  In particular, the cycle $\ell(1) \longrightarrow_{\tau_1} \ell(2) \longrightarrow_{\tau_2} \ell(1)$ is not apparent in any of the states.
\end{example}
This example reveals a fundamental problem when adapting ADCL for proving termination:
ADCL ensures that all reachable \emph{configurations} are covered, which is crucial for proving safety, but there are no such guarantees for all \emph{runs}.
Therefore, we think that adapting ADCL for proving termination requires major changes.

\section{Related Work and Experiments}
\label{sec:experiments}

We presented ADCL-NT, a variant of ADCL for proving non-termination.
The key insight is that tightly integrating techniques to detect
non-terminating transitions into ADCL allows for handling classes of TSs that are challenging for other techniques.
In particular, ADCL-NT can find non-terminating executions involving disjunctive transitions or complex patterns of transitions.
Moreover, it tightly couples the search for non-terminating configurations and the proof of their reachability, whereas other approaches usually separate these two steps.

\vspace{-0.5em}
\paragraph*{\bf Related Work}

There are many techniques to find certificates of non-termination \cite{larraz14,geometricNonterm,amram19,sttt22,lctrs-loops,loat}.
We could use any of them (they are black boxes for ADCL-NT).

Most non-termination techniques for TSs first search for non-ter\-mi\-nat\-ing configurations, and then prove their reachability \cite{larraz14,irankWST,t2-tool,jbc-nonterm}, or they extract and analyze \emph{lassos} \cite{geometricNonterm}.
In contrast, ADCL-NT tightly integrates the search for non-terminating configurations and reachability analysis.

Earlier versions of our tool \tool{LoAT} \cite{fmcad19,loat} also interleaved both steps using a technique akin to the state elimination method to transform finite automata to regular expressions.
This technique cannot handle disjunctions, and it is incomplete for reachability.
Hence, \tool{LoAT} is now solely based on ADCL-NT.

\vspace{-0.5em}
\paragraph*{\bf Implementation}

So far, our implementation in our tool \tool{LoAT} is restricted to integer arithmetic.
It uses the technique from \cite{loat} for acceleration and finding certificates of non-termination, the SMT solvers \tool{Z3} \cite{z3} and \tool{Yices} \cite{yices}, the recurrence solver \tool{PURRS} \cite{purrs}, and \tool{libFAUDES} \cite{faudes} to implement the automata-based redundancy check from \cite{cav23}.

\vspace{-0.5em}
\paragraph*{\bf Experiments}

To evaluate our implementation in \tool{LoAT}, we used the 1222 \emph{Integer Transition Systems} (ITSs) and the 335
\emph{{\normalfont{\pl{C}}} Integer Programs} from the \emph{Termination Problems
Database}~\cite{tpdb} used in \emph{TermComp} \cite{termcomp}.
The \pl{C} programs are small, hand-crafted examples that often require complex proofs.
The ITSs are significantly larger, as they were obtained from automatic transformations of \pl{C} or \pl{Java} programs.
Moreover, they contain a lot of ``noise'', e.g., branches where termination is trivial or variables that are irrelevant for (non-)termination.
Thus, they are well suited to test the scalability and robustness of the tools.

We compared our implementation (\tool{LoAT ADCL}) with other leading termina\-tion analyzers: \tool{iRankFinder}~\cite{irankWST,amram19}, \tool{T2}~\cite{t2-tool}, \tool{Ultimate}~\cite{Ultimate}, \tool{VeryMax}~\cite{larraz14,borralleras17}, and the previous version of \tool{LoAT}~\cite{loat} (\tool{LoAT '22}).
For \tool{T2}, \tool{VeryMax}, and \tool{Ultimate}, we took the versions of their last \emph{TermComp} participations (2015, 2019, and 2022).
For \tool{iRankFinder}, we used the configuration from the evaluation of \cite{loat}, which is tailored towards proving non-termination.
We excluded \tool{AProVE} \cite{tool-jar}, as it cannot prove
non-termination of ITSs, and it uses \tool{LoAT} and \tool{T2} as backends when analyzing \pl{C} programs.
Moreover, we excluded \tool{Ultimate} from the evaluation on ITSs, as it cannot parse them.
All experiments were run on \tool{StarExec} \cite{starexec} with $300$s wallclock timeout, $1200$s CPU timeout, and $128$GB memory limit per example.

\begin{table}
  \begin{center}
    \begin{tabular}{|c||c|c||c||c|c|c||c|c|}
      \hhline{~--------} \multicolumn{1}{c|}{} & \multicolumn{2}{c||}{No} & Yes & \multicolumn{3}{c||}{Runtime overall} & \multicolumn{2}{c|}{Runtime No} \\
      \hhline{~--------} \multicolumn{1}{c|}{}    & solved & unique & solved & average & median & timeouts & average & median \\
      \hhline{-========} \tool{LoAT ADCL}
                                               & 521    & 9      & 0      & 48.6~s  & 0.1~s  & 183      & 2.9~s   & 0.1~s  \\
      \hhline{---------} \tool{LoAT '22}     & 494    & 2      & 0      & 7.4~s   & 0.1~s  &   0      & 6.2~s   & 0.1~s  \\
      \hhline{---------} \tool{T2}           & 442    & 3      & 615    & 17.2~s  & 0.6~s  &  45      & 7.4~s   & 0.6~s  \\
      \hhline{---------} \tool{VeryMax}      & 421    & 6      & 631    & 28.3~s  & 0.5~s  &  30      & 30.5~s  & 14.5~s \\
      \hhline{---------} \tool{iRankFinder}  & 409    & 0      & 642    & 32.0~s  & 2.0~s  &  93      & 12.3~s  & 1.7~s  \\
      \hline
    \end{tabular}
  \end{center}
  \vspace{-1em}
\end{table}
\noindent
The table above shows the results for ITSs, where the column ``unique'' contains the number of examples that could be solved by
the respective tool, but no others.
It shows that \tool{LoAT ADCL} is the most powerful tool for proving non-termination of ITSs.
The main reasons for the improvement are that \tool{LoAT ADCL} builds upon a complete technique for proving reachability (in contrast to, e.g., \tool{LoAT '22}), and the close integration of non-termination techniques into a technique for proving reachability, whereas most competing tools separate these steps from each other.

If we only consider the examples where non-termination is proven, \tool{LoAT ADCL} is also the fastest tool.
If we consider all examples, then the \emph{average} runtime of \tool{LoAT ADCL} is significantly slower.
This is not surprising, as ADCL-NT does not terminate in general.
So while it is very fast in most cases (as witnessed by the very fast \emph{median} runtime), it times out more often than the other tools.

For \pl{C} integer programs, the best tools are very close (\tool{VeryMax}:
$103{\times}$No, \tool{LoAT ADCL}: $102{\times}$No, \tool{Ultimate}: $100{\times}$No).
Regarding runtimes, the situation is analogous to ITSs.
See \cite{website} for detailed results, more information about our evaluation, and a pre-compiled binary.
\tool{LoAT} is open-source and available on GitHub \cite{github}.

\bibliographystyle{splncs04}
\bibliography{refs,crossrefs,strings}

\providecommand{\noopsort}[1]{}
\begin{thebibliography}{10}
\providecommand{\url}[1]{\texttt{#1}}
\providecommand{\urlprefix}{URL }
\providecommand{\doi}[1]{https://doi.org/#1}

\bibitem{purrs}
Bagnara, R., Pescetti, A., Zaccagnini, A., Zaffanella, E.: \tool{PURRS}:
  Towards computer algebra support for fully automatic worst{-}case complexity
  analysis. CoRR  \textbf{abs/cs/0512056} (2005),
  \url{https://arxiv.org/abs/cs/0512056}

\bibitem{amram19}
Ben{-}Amram, A.M., Dom{\'{e}}nech, J.J., Genaim, S.: Multiphase-linear ranking
  functions and their relation to recurrent sets. In: SAS~'19. pp. 459--480.
  LNCS 11822 (2019). \doi{10.1007/978-3-030-32304-2\_22}

\bibitem{borralleras17}
Borralleras, C., Brockschmidt, M., Larraz, D., Oliveras, A.,
  Rodr{\'{\i}}guez{-}Carbonell, E., Rubio, A.: Proving termination through
  conditional termination. In: TACAS~'17. pp. 99--117. LNCS 10205 (2017).
  \doi{10.1007/978-3-662-54577-5\_6}

\bibitem{bozga09a}
Bozga, M., G{\^{\i}}rlea, C., Iosif, R.: Iterating octagons. In: TACAS~'09. pp.
  337--351. LNCS 5505 (2009). \doi{10.1007/978-3-642-00768-2\_29}

\bibitem{jbc-nonterm}
Brockschmidt, M., Str{\"{o}}der, T., Otto, C., Giesl, J.: Automated detection
  of non-termination and {{\texttt{{N}ull{P}ointer{E}xception}}}s for
  \tool{Java Bytecode}. In: FoVeOOS~'11. pp. 123--141. LNCS 7421 (2012).
  \doi{10.1007/978-3-642-31762-0\_9}

\bibitem{t2-tool}
Brockschmidt\noopsort{3}, M., Cook, B., Ishtiaq, S., Khlaaf, H., Piterman, N.:
  {\tool{T2}:} {T}emporal property verification. In: TACAS~'16. pp. 387--393.
  LNCS 9636 (2016). \doi{10.1007/978-3-662-49674-9\_22}

\bibitem{CHC-COMP}
{CHC Competition}, \url{https://chc-comp.github.io}

\bibitem{Ultimate}
Chen, Y., Heizmann, M., Leng{\'{a}}l, O., Li, Y., Tsai, M., Turrini, A., Zhang,
  L.: Advanced automata-based algorithms for program termination checking. In:
  PLDI~'18. pp. 135--150 (2018). \doi{10.1145/3192366.3192405}

\bibitem{irankWST}
Dom{\'{e}}nech, J.J., Genaim, S.: {\tool{iRankFinder}}. In: WST~'18. p.~83
  (2018), \url{https://wst2018.webs.upv.es/wst2018proceedings.pdf}

\bibitem{samir19}
Dom{\'{e}}nech\noopsort{1}, J.J., Gallagher, J.P., Genaim, S.: Control-flow
  refinement by partial evaluation, and its application to termination and cost
  analysis. Theory Pract.\ Log.\ Program.  \textbf{19}(5-6),  990--1005 (2019).
  \doi{10.1017/S1471068419000310}

\bibitem{yices}
Dutertre, B.: \tool{Yices} 2.2. In: CAV~'14. pp. 737--744. LNCS 8559 (2014).
  \doi{10.1007/978-3-319-08867-9\_49}

\bibitem{fmcad19}
Frohn, F., Giesl, J.: Proving non-termination via loop acceleration. In:
  FMCAD~'19. pp. 221--230 (2019). \doi{10.23919/FMCAD.2019.8894271}

\bibitem{acceleration-calculus}
Frohn, F.: A calculus for modular loop acceleration. In: TACAS~'20. pp. 58--76.
  LNCS 12078 (2020). \doi{10.1007/978-3-030-45190-5\_4}

\bibitem{sttt22}
Frohn, F., Fuhs, C.: A calculus for modular loop acceleration and
  non-termination proofs. Int.\ J.\ Softw.\ Tools Technol.\ Transf.
  \textbf{24}(5),  691--715 (2022). \doi{10.1007/s10009-022-00670-2}

\bibitem{loat}
Frohn, F., Giesl, J.: Proving non-termination and lower runtime bounds with
  \tool{LoAT} (system description). In: IJCAR~'22. pp. 712--722. LNCS 13385
  (2022). \doi{10.1007/978-3-031-10769-6\_41}

\bibitem{cav23}
Frohn, F., Giesl, J.: {ADCL}: {A}cceleration {D}riven {C}lause {L}earning for
  constrained {H}orn clauses. CoRR  \textbf{abs/2303.01827} (2023),
  \url{https://arxiv.org/abs/2303.01827}

\bibitem{github}
Frohn\noopsort{1}, F.: \tool{LoAT} on {G}it{H}ub (2023),
  \url{https://github.com/LoAT-developers/LoAT}

\bibitem{website}
Frohn\noopsort{1}, F., Giesl, J.: Empirical evaluation of ``{P}roving
  non-termination by {A}cceleration {D}riven {C}lause {L}earning'' (2023),
  \url{https://loat-developers.github.io/adcl-nonterm-eval}

\bibitem{tool-jar}
Giesl, J., Aschermann, C., Brockschmidt, M., Emmes, F., Frohn, F., Fuhs, C.,
  Hensel, J., Otto, C., Plücker, M., Schneider{-}Kamp, P., Ströder, T.,
  Swiderski, S., Thiemann, R.: Analyzing program termination and complexity
  automatically with \tool{AProVE}. {J.\ Autom.\ Reasoning}  \textbf{58}(1),
  3--31 (2017). \doi{10.1007/s10817-016-9388-y}

\bibitem{termcomp}
Giesl, J., Rubio, A., Sternagel, C., Waldmann, J., Yamada, A.: The termination
  and complexity competition. In: TACAS~'19. pp. 156--166. LNCS 11429 (2019).
  \doi{10.1007/978-3-030-17502-3\_10}

\bibitem{larraz14}
Larraz\noopsort{1}, D., Nimkar, K., Oliveras, A., Rodríguez{-}Carbonell, E.,
  Rubio, A.: Proving non-termination using {M}ax-{SMT}. In: CAV~'14. pp.
  779--796. LNCS 8559 (2014). \doi{10.1007/978-3-319-08867-9\_52}

\bibitem{geometricNonterm}
Leike, J., Heizmann, M.: Geometric nontermination arguments. In: TACAS~'18. pp.
  266--283. LNCS 10806 (2018). \doi{10.1007/978-3-319-89963-3\_16}

\bibitem{faudes}
{\tool{libFAUDES} Library}, \url{https://fgdes.tf.fau.de/faudes/index.html}

\bibitem{lctrs-loops}
Nishida, N., Winkler, S.: Loop detection by logically constrained term
  rewriting. In: VSTTE '18. pp. 309--321. LNCS 11294 (2018).
  \doi{10.1007/978-3-030-03592-1\_18}

\bibitem{z3}
\noopsort{Moura}{de Moura}, L., Bj{\o}rner, N.: \tool{Z3}: An efficient {SMT}
  solver. In: TACAS\ '08. pp. 337--340. LNCS 4963 (2008).
  \doi{10.1007/978-3-540-78800-3\_24}

\bibitem{oeis-thue-morse}
{OEIS Foundation Inc.}: Thue-{M}orse sequence. The {O}n-{L}ine {E}ncyclopedia
  of {I}nteger {S}equences, published electronically at
  \url{https://oeis.org/A010060}

\bibitem{oeis-thue-morse-difference}
{OEIS Foundation Inc.}: First differences of {T}hue-{M}orse sequence. The
  {O}n-{L}ine {E}ncyclopedia of {I}nteger {S}equences (1999), published
  electronically at \url{https://oeis.org/A029883}

\bibitem{starexec}
Stump, A., Sutcliffe, G., Tinelli, C.: \tool{StarExec}: {A} cross-community
  infrastructure for logic solving. In: IJCAR~'14. pp. 367--373. LNCS 8562
  (2014). \doi{10.1007/978-3-319-08587-6\_28}

\bibitem{tpdb}
Termination {P}roblems {D}ata {B}ase ({TPDB}),
  \url{https://termination-portal.org/wiki/TPDB}

\end{thebibliography}

\report{
\begin{appendix}
  \appendixproofsection{Missing Proofs}\label{sec:MissingProofs}
\appendixproof*{thm:correct}

\noindent
The following variant of \Cref{thm:correct} for $\leadsto_\nt$ can be proven analogously.
\begin{lemma}
  \label{lem:run}
  If $\TT \leadsto_\nt^* (\SS,\vec{\tau},\vec{B})$ and $\vec{\tau}$ is non-empty and safe,
  then  $\cond(\vec{\tau}) \not \equiv_\AA \bot$ and
  ${\longrightarrow_{\vec{\tau}}} \subseteq {\longrightarrow^+_\TT}$.

\end{lemma}

\appendixproof*{thm:sound-nt}
\appendixproof*{lem:lang}
\appendixproof*{thm:incomplete}

\end{appendix}}

\end{document}